\documentclass[11pt]{article}
\usepackage{amsthm}
\usepackage[unicode]{hyperref}
\usepackage[utf8]{inputenc}
\usepackage{amsmath}
\usepackage{amssymb} 
\usepackage{mathtools}
\usepackage{thmtools} 
\usepackage{thm-restate}
\usepackage{xcolor,xspace}
\usepackage{nicefrac}
\usepackage{subfigure}
\usepackage{tabularx}
\newcolumntype{C}[1]{>{\hsize=#1\hsize\centering\arraybackslash}X}
\usepackage{multirow}
 \usepackage{enumerate}  
\usepackage{framed}

\usepackage[nameinlink,capitalize]{cleveref}

\newtheorem{theorem}{Theorem}[section]
\newtheorem{definition}[theorem]{Definition}
\newtheorem{lemma}[theorem]{Lemma}
\newtheorem{proposition}[theorem]{Proposition}

\usepackage{algpseudocode}
\usepackage{algorithm}

\newcommand{\alg}[1]{\textsc{#1}}

\newcommand{\CONGEST}{\ensuremath{\mathsf{CONGEST}}\xspace}
\newcommand{\LOCAL}{\ensuremath{\mathsf{LOCAL}}\xspace}
\newcommand{\congest}{\CONGEST}
\DeclareMathOperator*{\E}{\mathbb{E}}

\DeclareMathOperator{\poly}{poly}

\newcommand{\bbN}{{\mathbb N}}

\newcommand{\bbR}{{\mathbb R}}

\newcommand{\calC}{{\mathcal C}}
\newcommand{\calF}{{\mathcal F}}

\newcommand{\good}{{\mathrm{good}}}

\newcommand{\knuthupuparrow}{\mathbin{\uparrow\uparrow}}

\DeclarePairedDelimiter\card{\lvert}{\rvert}
\DeclarePairedDelimiter\norm{\lVert}{\rVert}
\DeclarePairedDelimiter\set{\lbrace}{\rbrace}
\DeclarePairedDelimiter\event{[}{]}
\DeclarePairedDelimiter\range{[}{]}

\DeclarePairedDelimiter\parens{\lparen}{\rparen}

\newcommand{\zo}{\{0,1\}}

\begin{document}

\title{Superfast Coloring in CONGEST via Efficient Color Sampling}

\author{Magn\'us M. Halld\'orsson\thanks{ICE-TCS \& Department of Computer Science, Reykjavik University, Iceland. Partially supported by Icelandic Research Fund grant 174484-051.}
\and Alexandre Nolin\footnotemark[1]}

\date{\today}

\maketitle

\begin{abstract}
 We present a procedure for efficiently sampling colors in the {\congest} model. It allows nodes whose number of colors exceeds their number of neighbors by a constant fraction to sample up to $\Theta(\log n)$ semi-random colors unused by their neighbors in $O(1)$ rounds, even in the distance-2 setting. 
 This yields algorithms with $O(\log^* \Delta)$ complexity for different edge-coloring, vertex coloring, and distance-2 coloring problems, matching the best possible.
 In particular, we obtain an $O(\log^* \Delta)$-round {\CONGEST} algorithm for $(1+\epsilon)\Delta$-edge coloring when $\Delta \ge \log^{1+1/\log^*n} n$, and a poly($\log\log n$)-round algorithm for $(2\Delta-1)$-edge coloring in general.
The sampling procedure is inspired by a seminal result of Newman in communication complexity.
\end{abstract}

\section{Introduction}

The two primary models of locality, \LOCAL and \CONGEST, share most of the same features: the nodes are connected in the form of an undirected graph, time proceeds in synchronous rounds, and in each round, each node can exchange different messages with each of its neighbors. The difference is that the messages can be of arbitrary size in \LOCAL, but only logarithmic in \CONGEST. A question of major current interest is to what extent message sizes matter in order to achieve fast execution.

Random sampling is an important and powerful principle with extensive applications to distributed algorithms. In its basic form, the nodes of the network compute their random samples and share it with their neighbors in order to reach collaborative decisions. When the samples are too large to fit in a single \CONGEST message, then the \LOCAL model seems to have a clear advantage. The goal of this work is to overcome this handicap and derive equally efficient \CONGEST algorithms, particularly in the context of coloring problems.

Graph coloring is one of the most fundamental topics in distributed computing. In fact, it was the subject of the first work on distributed graph algorithms by Linial \cite{linial92}. The task is to either color the \emph{vertices} or the \emph{edges} of the underlying communication graph $G$ so that adjacent vertices/edges receive different colors. The most basic distributed coloring question is to match what is achieved by a simple centralized algorithm that colors the vertices/edges in an arbitrary order. Thus, our primary focus is on the $(\Delta+1)$-vertex coloring and the $(2\Delta-1)$-edge coloring problems, where $\Delta$ is the maximum degree of $G$.

Randomized distributed coloring algorithms are generally based on sampling colors from the appropriate domain. The classical and early algorithms for vertex coloring, e.g.\ \cite{johansson99,BEPS16}, involve sampling individual colors and operate therefore equally well in \CONGEST. The more recent fast coloring algorithms, both for vertex \cite{SW10,EPS15,HSS16,CLP20} and edge coloring \cite{EPS15}, all involve a technique of Schneider and Wattenhofer \cite{SW10} that uses samples of up to logarithmic number of colors. In fact, there are no published sublogarithmic algorithms (in $n$ or $\Delta$) for these coloring problems in \CONGEST, while there are now $\poly(\log\log n)$-round algorithms \cite{EPS15,CLP20,GGR20} in \LOCAL. A case in point is the $(2\Delta-1)$-edge coloring problem when $\Delta = \log^{1+\Omega(1)} n$, which can be solved in only $O(\log^* n)$ \LOCAL rounds \cite{EPS15}. The bottleneck in \CONGEST is the sampling size of the Schneider-Wattenhofer protocol. 

We present here a technique for sampling a logarithmic number of colors and communicating them in only $O(1)$ \CONGEST rounds. 
We apply the technique to a number of coloring problems, allowing us to match in \CONGEST the best complexity known in \LOCAL.

The sampling technique is best viewed as making random choices with a limited amount of randomness.
This is achieved by showing that sampling within an appropriate subfamily of all color samples can retain some of the useful statistical properties of a fully random sample.
It is inspired by Newman's theorem in communication complexity \cite{Newman91}, where
dependence on shared randomness is removed through a similar argument.

We apply the sampling technique to a number of coloring problems where the nodes/edges to be colored have a large \emph{slack}: the number of colors available exceeds by a constant fraction the number of neighbors. 
We particularly apply the technique to settings where the maximum degree $\Delta$ is superlogarithmic (we shall assume $\Delta=\Omega(\log^{1+1/\log^*n} n)$).

We obtain a superfast $O(\log^* \Delta)$-round algorithm for $(2\Delta-1)$-edge coloring when $\Delta=\Omega(\log^{1+1/\log^*n} n)$.
Independent of $\Delta$, we obtain a $\poly(\log\log n)$-round algorithm.
This shows that coloring need not be any slower in {\CONGEST} than in {\LOCAL}.

We obtain similar results for vertex coloring, for the same values of $\Delta$ ($\Delta=\Omega(\log^{1+1/\log^*n} n)$).
We obtain an $O(\log^* \Delta)$-round algorithm for $(1+\epsilon)\Delta$-coloring, for any $\epsilon > 0$.
For graphs that are locally sparse (see Sec.~2 for definition), this gives a $(\Delta+1)$-coloring in the same time complexity.
Matching results also hold for the \emph{distance-2} coloring problem, where nodes within distance 2 must receive different colors.

\subsection{Related Work}
The literature on distributed coloring is vast and we limit this discussion to work that is directly relevant to ours, primarily randomized algorithms.

An edge coloring of a graph $G$ corresponds to a vertex coloring of its line graph, whose maximum degree is $2\Delta(G)-2$. Therefore, \LOCAL algorithms for $(\Delta+1)$-vertex coloring yield $(2\Delta-1)$-edge coloring in the same time. Since line graphs have a special structure, edge coloring often allows for either faster algorithms or fewer number of colors.
For \CONGEST, the situation is different:
Because of capacity restrictions, no single node can expect to learn the colors of all edges adjacent to a given edge.
In fact, there are no published results on efficient edge-coloring algorithms in \CONGEST, to the best of our knowledge\footnote{Fischer, Ghaffari and Kuhn \cite{FGK17} suggest in a footnote that their edge coloring algorithms, described and proven in \LOCAL, actually work in \CONGEST. It does not hold for their randomized edge-coloring result, which applies the algorithm of \cite{EPS15}.}. 

A classical simple (probably folklore) algorithm for vertex coloring is for each vertex to pick in each round a color uniformly at random from its \emph{current palette}, the colors that are not used on neighbors. Each node can be shown to become colored in each round with constant probability and thus this procedure completes in $O(\log n)$ rounds, w.h.p.\ \cite{johansson99}. 
In fact, each round of this procedure reduces w.h.p.\ the \emph{uncolored degree} of each vertex by a constant factor, as long as the degree is $\Omega(\log n)$~\cite{BEPS16}. 
Within $O(\log \Delta)$ rounds we are then in the setting where the maximum uncolored degree of each node is logarithmic. 
This algorithm works also in \CONGEST for node coloring, and as well for edge coloring in \LOCAL, but 
does not immediately work for edge coloring in \CONGEST, since it is not clear how to select a color uniformly at random from the palette of an edge.

Color sampling algorithms along a similar vein have also been studied for edge coloring \cite{PS97,GP97,DGP98}, all running in $O(\log n)$ \LOCAL rounds in general. Panconesi and Srinivasan \cite{PS97} showed that one of the most basic algorithms finds a $(1.6\Delta+\log^{2+\Omega(1)} n)$-edge coloring. 
Grable and Panconesi \cite{GP97} showed that $O(\log\log n)$ rounds suffice when $\Delta = n^{\Omega(1/\log \log n))}$.
Dubhashi, Grable and Panconesi \cite{DGP98} proposed an algorithm based on the R\"odl nibble technique, where only a subset of the edges try a color in each round, and showed that it finds a $(1+\epsilon)\Delta$-edge coloring, when $\Delta = \omega(\log n)$.

Sublogarithmic round vertex coloring algorithms have two phases, where the first phase is completed once the \emph{uncolored degree} of the nodes is low (logarithmic or polylogarithmic). 
Barenboim et al.\ \cite{BEPS16} showed that within $O(\log\log n)$ additional rounds, the graph is \emph{shattered}: each connected component (induced by the uncolored nodes) is of polylogarithmic size.
The default approach is then to apply fast deterministic algorithms. With recent progress on network decomposition  \cite{RG19,GGR20}, as well as fast deterministic coloring algorithms \cite{GK20}, the low degree case can now be solved in $\poly(\log\log n)$ rounds.

Recent years have seen fast \LOCAL coloring algorithms that run in sublogarithmic time. These methods depend crucially on a random sampling method of Schneider and Wattenhofer \cite{SW10} where each node picks as many as $\log n$ colors at a time.
The method works when each node has large \emph{slack}; i.e., when the number of colors in the node's palette is a constant fraction larger than the number of neighbors (competing for those colors). This holds in particular when computing a $(1+\epsilon)\Delta$-coloring, for some $\epsilon > 0$, which they achieve in $O(\log^* \Delta)$ rounds, when $\Delta \ge \log^{1+\Omega(1)} n$.

In the $(\Delta+1)$-node coloring and $(2\Delta-1)$-edge coloring problems, the nodes do not have any slack \emph{a priori}. 
It turns out that such slack can sometimes be generated by a single round of color guessing. Suppose the graph is triangle free, or more generally, \emph{locally sparse}, meaning that the induced subgraph of each node has many non-adjacent pairs of nodes. Then, when each node tries random color, each pair of non-adjacent common neighbors of $v$ has a fair chance of being colored with the same color, which leads to an increase in the slack of $v$. As shown by Elkin, Pettie and Su \cite{EPS15} (with a longer history in graph theory, tracing back at least to Reed \cite{Reed98}), locally sparse graphs will have slack $\Omega(\Delta)$ after this single color sampling round.
Line graphs are locally sparse graphs, and thus we obtain this way a $O(\log^* \Delta)$-round algorithm for $(2\Delta-1)$-edge coloring \cite{EPS15}, for $\Delta \ge \Delta^{1+\Omega(1)}$. They further obtain a $(1+\epsilon)\Delta$-edge list coloring in the same time frame, using the nibble technique of \cite{PS97}.

This fast coloring of locally sparse graphs is also useful in $(\Delta+1)$-vertex coloring. Both the first sublogarithmic round algorithm of Harris, Schneider, Su \cite{HSS16} and the current fastest algorithm of Chang, Li, and Pettie \cite{CLP20} partition the graph into a sparse and a dense part, color the sparse part with a variation of the method of \cite{SW10}, and synchronize the communication within each cluster of the dense part to achieve fast coloring.

A \emph{distance-2 coloring} is a vertex coloring such that nodes within distance at most 2 receive different colors. This problem in \CONGEST shares a key property with edge coloring: nodes cannot obtain a full knowledge of their available palette, but they can try a color by asking their neighbors. 
A recent $(\Delta^2+1)$-distance-2 coloring algorithm of \cite{HKMN20_fullversion} that runs in $O(\log \Delta) + \poly(\log\log n)$ \CONGEST rounds can be used to compute $(2\Delta-1)$-edge colorings in the same time complexity.

\section{Intuition and preliminaries}
\label{sec:intuition-preliminaries}

Existing $O(\log^* \Delta)$ algorithms for the different coloring problems in {\LOCAL} such as those by Schneider and Wattenhofer~\cite{SW10} all involve sampling several colors in a single round. In such algorithms, the nodes try colors in a way that guarantees each color an independent, $\Omega(1)$ probability of success. While this probability of success is a given when all nodes try a single color, having each node try several colors in any given round could create more conflicts between colors and reduce the probability of success of any given one.

This issue is usually solved using 
\emph{slack}, the difference between the number of colors unused by the neighbors of a node and how many of its neighbors are still uncolored. Put another way, slack is the number of colors that is guaranteed to be left untouched by your neighbors for all possible choices of your currently uncolored neighbors. Slack is a given when we allow more colors than each node has neighbors, and is otherwise easily generated in a locally sparse graph.

If the nodes are all able to try $\Theta(\log n)$ colors in $O(1)$ rounds, and all colors have an independent, $\Omega(1)$ probability of success, $O(1)$ rounds suffice to color all nodes w.h.p. However, this is usually not immediately possible, unless all nodes have a large amount of slack from the beginning.
The $O(\log^* n)$ algorithms work through increasing the ratio of slack to uncolored degree, trying more and more colors as this ratio increases, allowing nodes to try $\Theta(\log n)$ colors each with constant probability over the course of $O(\log^* n)$ rounds. The speed comes from the fact that slack never decreases but the uncolored degree of the edges decreases with exponentially increasing speed as the nodes try more and more colors.

However, all these algorithms have nodes send $\Theta(\log n)$ colors during the algorithm's execution, which requires $\Theta(\log n \cdot \log \Delta)$ bits, i.e., a minimum of $\Theta(\log \Delta)$ {\congest} rounds. Our algorithms will also involve having each node try up to $\Theta(\log n)$ colors, but without transmitting $\Theta(\log n)$ arbitrary colors.

\subsection{Sampling colors with shared randomness}

While $\Theta(\log n \cdot \log \Delta)$ bits are needed to describe an arbitrary choice of $\Theta(\log n)$ colors in a color space of size $\Theta(\Delta)$, being able to describe any choice of $\Theta(\log n)$ colors can be unnecessary. To get intuition about this, consider the setting where all nodes have access to a shared source of randomness. When trying random colors, the edges do not care about which specific set of colors they are trying, all that matters is that the colors they try are random and independent of what other nodes are trying.

With a shared source of randomness, instead of sending $\log \Delta$ bits to specify a color, a node can use the shared random source as a source of random colors and send indices of colors in the random source.
If each random color has a chance $\geq p$ of having the properties needed to be tried, the index of the first satisfactory color will be of expected value $O(1/p)$ and only take $O(\log(1/p))$ bits to communicate. The nodes can also use $O(\log n)$ bits to indicate which of the first $O(\log n)$ colors in the random source they find satisfactory and decide to try.
This technique allows the edges to sample $\Theta(p \log n)$ colors in a single round of \CONGEST. The choices made by nodes are made independent 
by having the nodes use disjoint parts of the shared randomness (for example, each node might only use the bits at indices equal to its ID modulo $n$).
{This type of saving in the communication based on a shared source of randomness appears in several places in communication complexity, in particular in~\cite{HastadW07} where it is used with the Disjointness problem, and in the folklore protocol for Equality (e.g., Example~3.13 in~\cite{KN97}).}

It is crucial in the above argument that all nodes have access to a shared source of randomness, as messages making references to the shared randomness lose their meaning without it.
Our goal will now be to remove this need for a shared source of randomness, taking inspiration from Newman's Theorem in communication complexity~\cite{Newman91} (Theorem 3.14 in~\cite{KN97}, Theorem 3.5 in~\cite{RY2020}). 
It is not an application of it, however, as contrary to the $2$-party communication complexity setting, distributing a common random seed to all parties would require many rounds in our context, and the success of any node trying one or more colors is interrelated with the random choices of up to $\Delta+1$ parties. Our contribution is best understood as replacing a fully random sample of colors by a pseudorandom one with appropriate statistical guarantees, whose proof of existence resembles the proof of Newman's Theorem.
We do so in 
Section~\ref{sec:trick}, 
and give multiple applications of this result in subsequent sections.

\subsection{Tools and notation}

Our results rely heavily on the existence of a family of sets with the right properties, whose existence we prove by a probabilistic argument. We make frequent use of the Chernoff-Hoeffding bounds in this proof, as well as in other parts of the paper. We use a version of the bounds that holds for \emph{negatively associated} random variables.

\begin{definition}[Negative association]
\label{def:negative-association}
The random variables $X_1,\ldots,$ $X_n$ are said to be \emph{negatively associated} if for all disjoint subsets $I,J \subseteq [n]$ and all non-decreasing functions $f$ and~$g$,
\[\E\event{f(X_i,i\in I) \cdot g(X_j,j\in J)} \leq \E\event{f(X_i,i\in I)} \cdot \E\event{g(X_j,j\in J)}\]
\end{definition}

\begin{lemma}[Chernoff-Hoeffding bounds]
\label{lem:chernoff-hoeffding}
Let $X_1,\ldots,X_n$ be $n$ negatively associated random variables in $[0,1]$, $X := \sum_{i=1}^n X_i$ their sum, and let the expectation of $X$ satisfy $\mu_L \leq \E[X] \leq \mu_H$. For $0 < \epsilon < 1$:
\begin{align}
    \Pr\event{X > (1+\epsilon)\mu_H} \leq \exp \left ( - \frac {\epsilon^2} 3 \mu_H \right ),\\
    \Pr\event{X < (1-\epsilon)\mu_L} \leq \exp \left ( - \frac {\epsilon^2} 2 \mu_L \right ).
\end{align}
\end{lemma}

Negative association is a somewhat complicated-looking property but the property holds in simple scenarios. In particular it holds for balls and bins experiments~\cite{DubhashiRanjan98,DP09}, such as when the random variables $X_1,\ldots,X_n$ correspond to sampling $k$ elements out of $n$ (i.e., when the random variables satisfy $\Pr\event{X_i = v_i, \forall i \in [n]} = 1/\binom{n}{k}$ for all $v\in \zo^n, \norm{v}_1=k$). It also encompasses the usual setting where $X_1,\ldots,X_n$ are independent.

For ease of notation, we will use the shorthand $[a,b]k$ to denote the interval $[a\cdot k,b\cdot k]$, $[a..b]$ to denote the set $\set{a,\ldots,b}$, and $[k]$ to denote the set $\set{1,\ldots,k}$.

Throughout the paper we describe algorithms that try an increasing number of colors in a single round. This increase is much faster than exponential and we use Knuth's up-arrow notation to denote it. In fact, the increase is as fast as the inverse of $\log^*$, which already gives a sense of why our algorithms run in $O(\log^* n)$ rounds.

\begin{definition}[Knuth's up-arrow notation for tetration]
    For $a\in \bbR, b\in \bbN$, $a \knuthupuparrow b$ represents the \emph{tetration} or \emph{iterated exponentiation} of $a$ by $b$, defined as:
    \[a \knuthupuparrow b = \begin{cases}1 & \textrm{if }b=0\\a^{a \knuthupuparrow (b-1)} & \textrm{otherwise}\end{cases}\]
\end{definition}
Throughout the paper, as we work on a graph $G(V,E)$ of vertices $V$ and edges $E$, we denote by $n$ the number of vertices and by $\Delta$ the maximum degree of the graph. The degree of a vertex is denoted by $d(v)$, its uncolored degree (how many of its neighbors are uncolored) by $d^*(v)$. The sparsity of $v$ (\cref{def:sparsity}) is denoted by $\zeta(v)$, the palette of $v$ (the set of colors not yet used by one of $v$'s neighbors) by $\psi_v$, and its slack $s(v)$ is defined as $s(v)=\card{\psi_v}-d^*(v)$. Whenever we consider an edge-coloring problem, we will often work on the line graph and add an $L$ subscript to indicate that we consider the same quantities but on $L(G)$: the maximum degree of this graph is $\Delta_L=2\Delta-1$, the degree of an edge is denoted by $d_L(e)$, and so on.

\begin{definition}[Sparsity]
\label{def:sparsity}
    Let $v$ be a node in the graph $G(V,E)$ of maximum degree $\Delta$, and let $E[N(v)]$ the set of edges between nodes of $v$'s neighborhood $N(v)$. The sparsity of $v$ is defined as:
    \[\zeta(v) = \frac 1 {\Delta} \cdot \left ( \binom{\Delta}{2} - \card*{E[N(v)]} \right ) \]
\end{definition}

The sparsity is a measure of how many edges are missing out of all the edges that could exist in the neighborhood of a node. As immediate property, $\zeta(v)$ is a rational number in the range $[0,(\Delta-1)/2]$. A value close to $0$ indicates a very dense neighborhood (a value of exactly $0$ indicates that $v$'s neighbors form a clique of $\Delta$ nodes) while a value close to $(\Delta-1)/2$ indicates the opposite, that $v$'s neighborhood is sparse (a value of $(\Delta-1)/2$ means that no two neighbors of $v$ are connected to one another). A graph is said to be \emph{$(1-\epsilon)$-locally sparse} iff its vertices are all of sparsity at least $\epsilon\Delta$. A vertex $v$ of sparsity $\zeta$ is equivalently said to be $\zeta$-sparse.

Sparsity is of interest here for two reasons: first, because we know from a result of~\cite{EPS15} that nodes receive slack proportional to their sparsity w.h.p.\ in just one round of all nodes trying a random color if $\zeta(v) \in \Omega(\log n)$ (\cref{prop:slack-lemma}), and second because the line graph is sparse by construction (\cref{prop:line-graph-sparsity}), and therefore generating slack in it follows directly from \cref{prop:slack-lemma}.

\begin{proposition}[\cite{EPS15}, Lemma 3.1]
    \label{prop:slack-lemma}
    Let $v$ be a vertex of sparsity $\zeta$ and let $Z$ be the slack of $v$ after trying a single random color. Then, 
    \[\Pr[Z \le \zeta/(4 e^3)] \le e^{-\Omega(\zeta)}.\]
\end{proposition}

\begin{proposition}
    \label{prop:line-graph-sparsity}
    A node $e$ of the line graph $L(G)$ (i.e., an edge of $G$) has degree $d_L(e)$ at most $\Delta_L=2(\Delta-1)$, and the number of edges in its neighborhood $E_{L(G)}[N(e)\setminus \{e\}]$ is at most $(\Delta-1)^2$, meaning $e$ is $(\Delta-2)/2$-sparse, i.e., $(\Delta_L-2)/4$-sparse.
\end{proposition}

\section{Efficient color sampling with representative sets}
\label{sec:trick}

We now introduce the tool that will allow us to sample and communicate $\Theta(\log n)$ colors in $O(1)$ {\congest} rounds with the right probabilistic guarantees. Let $s$ be the number of elements we sample and $k$ the size of the universe to be sampled from.
If our goal was to be able to sample all random subsets of $[k]$ of size $s$, we would need $\log \binom{k}{s}$ bits to communicate our choice of subset.
But our goal is to communicate less than this amount, so we instead consider a family of $s$-sized subsets of $[k]$ such that picking one of those subsets at random has some of the probabilistic properties of sampling an $s$-sized subset of $[k]$ uniformly at random. The family is much smaller that the set of all possible $s$-sized subsets of $[k]$, which allows us to communicate a member of it in much less than $\log \binom{k}{s}$ bits. We call the family of subsets a \emph{representative family}, made of \emph{representative sets}, and the probabilistic properties we maintain are essentially that:
\begin{itemize}
    \item Every element of $[k]$ is present in about the same number of sets.
    \item For any large enough subset $T$ of $[k]$, a random representative set intersects $T$ in about the same number of elements as a fully random $s$-sized set would.
\end{itemize}
Crucially, the second property holds for a large enough arbitrary $T$, so we will be able to apply it even as $T$ is dependent on the choices of other nodes in the graph as long as the representative set is picked independently from $T$. $T$ will typically be the palette of a node or edge, or the set of colors not tried by any neighbors of a node or edge. Being able to just maintain the two properties above is enough to efficiently adapt many {\LOCAL} algorithms that rely on communicating large subsets of colors to the {\congest} setting.

\begin{definition}[Representative sets]
\label{def:representative-sets}
    Let $U$ be a universe of size $k$. A family $\calF = \{S_1,\ldots,S_t\}$ of $s$-sized sets is said to be an $(\alpha,\delta,\nu)$-representative family iff:
    \begin{align}
        \forall T \subseteq U, \card{T} \geq \delta k:&\quad
        \Pr_{i\in[t]}\event*{\frac{\card{S_i \cap T}}{\card{S_i}} \in [1-\alpha,1+\alpha]\frac{\card{T}}{k}} \geq (1 - \nu), \label{eq:rep-set-large-T}
        \\
        \forall T \subseteq U, \card{T} < \delta k:&\quad
        \Pr_{i\in[t]}\event*{\frac{\card{S_i \cap T}}{\card{S_i}} \leq (1+\alpha)\delta } \geq (1 - \nu),\label{eq:rep-set-small-T}
        \\
        \forall u \in U:&\quad
        \Pr_{i\in[t]}\event*{u\in S_i} \in [1 - \alpha,1+\alpha] \frac{s \cdot t}{k}.
        \label{eq:rep-set-single-elem}
    \end{align}
\end{definition}

We show in \cref{lem:representative-sets} that such families exist for some appropriate choices of parameters. The proof of this result, which relies on the probabilistic method, takes direct inspiration from Newman's Theorem \cite{Newman91}.

\begin{lemma}[Representative sets exist]
\label{lem:representative-sets}
    Let $U$ be a universe of size $k$. For any $\alpha,\delta,\nu >0$, 
    there exists an $(\alpha,\delta,\nu)$-representative family $(S_i)_{i\in[t]}$ of $t \in O(k/\nu + k \log (k))$ subsets, each of size $s \in O(\alpha^{-2}\delta^{-1}\log(1/\nu))$.
\end{lemma}
\begin{proof}
    Our proof is probabilistic: we show that Equations~\ref{eq:rep-set-large-T}, \ref{eq:rep-set-small-T} and~\ref{eq:rep-set-single-elem} all hold with non-zero probability when picking sets at random. We first study the probability that Equations~\ref{eq:rep-set-large-T} and~\ref{eq:rep-set-small-T} hold, and then the probability that Equation~\ref{eq:rep-set-single-elem} holds.

    Consider any set $T\subseteq U$ of size $\geq \delta k$. Pick a random set $S \subseteq U$ of size~$s$. The intersection of $S$ and $T$ has expected size $\E_S[\card{S\cap T}] = \frac{\card{T}}{k}s$. Let us say that $S$ has an \emph{unusual} intersection with $T$ if its size is outside the $[1-\alpha,1+\alpha]\frac{\card{T}}{k}s$ range. By Chernoff with negative dependence,
    \[\Pr_S \left[ \card{S\cap T} \not \in [1-\alpha,1+\alpha]\frac{\card{T}}{k}s \right] \leq 2e^{-s\alpha^2\frac{\card{T}}{3k}} \leq 2e^{-\frac{\alpha^2 \delta}{3}s}.\]
    
    This last quantity also bounds the probability that $\card{S\cap T} > (1+\alpha)\delta s$ when $\card{T} < \delta k$, which we also consider as an unusual intersection.
    
    Pick $t$ sets $S_1,\ldots,S_t$ of size $s$ at random independently from each other, let $X_i$ be the event that the i$^{th}$ set $S_i$ unusually intersects $T$. By Chernoff, the probability that more than $4 t \cdot \exp \left ({-\frac{\alpha^2\delta}{3}s}\right)$ of the sets unusually intersect $T$ is:
    \[\Pr_{S_1\ldots S_t}\left[ \sum_i X_i > 4 t \cdot e^{-\frac{\alpha^2\delta}{3}s} \right] \leq e^{- \frac {t} {3} \cdot \exp\left(-\frac{\alpha^2\delta}{3}s\right)}\]
    
    There are less than $2^{k}$ subsets of $U$. Therefore, the probability that there exists a set $T$ such that out of the $t$ sampled sets $S_1\ldots S_t$, more than $4 t \cdot \exp \left ({-\frac{\alpha^2\delta}{3}s}\right)$ have an unusual intersection with $T$,
    is at most:
    \begin{align*}
        2^{k} \cdot e^{- \frac {t} {3} \cdot \exp\left(-\frac{\alpha^2\delta}{3}s\right)} &=\exp\left(k\cdot \ln(2) - \frac {t} {3} \cdot \exp\left(-\frac{\alpha^2\delta}{3}s\right)\right)
    \end{align*}
    
    This last quantity is an upper bound on the probability that one of Equations~\ref{eq:rep-set-large-T} and~\ref{eq:rep-set-small-T} does not hold. Let us now similarly bound  the probability that Equation~\ref{eq:rep-set-single-elem} does not hold.
    
    For any $u\in U$, the probability that a random $s$-sized subset of $U$ contains $u$ is $s/k$. Let $X_i$ be the event that our i$^{th}$ random set $S_i$ contains $u$, we have:
    \[\Pr_{S_1\ldots S_t}\event*{ \sum_i X_i \not \in [1-\alpha,1+\alpha] \frac {s\cdot t}{k} } \leq 2e^{- \alpha^2\frac{s\cdot t}{3k}}\]
    
    Therefore the probability that Equation~\ref{eq:rep-set-single-elem} does not hold, i.e., that there exists an under- or over-represented element $u \in U$ in our $t$ randomly picked sets, is less than $2k\cdot e^{- \alpha^2\frac{s\cdot t}{3k}}$. The probability that one of Equations~\ref{eq:rep-set-large-T}, \ref{eq:rep-set-small-T}, and~\ref{eq:rep-set-single-elem} does not hold is at most:
    \[\exp\left(k\cdot \ln(2) - \frac {t} {3} \cdot \exp\left(-\frac{\alpha^2\delta}{3}s\right)\right) + \exp\left(\ln(2k) - \alpha^2\frac{s\cdot t}{3k}\right)\]
    
    We now pick the right values for $s$ and $t$ such that: first,  this last probability is less than $1$ and, therefore, a family with all the above properties exist; second, the fraction of sets $S_i$ with the wrong intersection is less than $\nu$ for all $T$.
    
    The fraction of bad sets is guaranteed to be less than $\nu$ if $4 \cdot e^{-\frac{\alpha^2\delta}{3}s} \leq \nu$, which is achieved with $s \geq \ln(4/\nu) \cdot \frac {3}{\alpha^2\delta}$. We take $s$ to be this last value rounded up, i.e., we have $s \in O(\alpha^{-2}\delta^{-1}\log(1/\nu))$. For $t$, we pick it satisfying $t > 3(k\cdot\ln(2)+1)\cdot \exp\left(\frac{\alpha^2\delta}{3}s\right)$ and $t > \frac{3k\cdot(\ln(2k)+1)}{\alpha^2 \cdot s}$, that is, we can pick $t$ of order $\Theta\left( {k} / {\nu} + k \log(k)\right)$ and satisfy all properties with non-zero probability, implying the existence of the desired representative family.
\end{proof}

\section{\texorpdfstring{$(1+\epsilon)\Delta$}{(1+ε)Δ}-vertex coloring}
\label{sec:d1-coloring}

For ease of exposition, we start by applying our techniques in a relatively simple setting before moving on to more complex ones. As many elements are similar between the different settings we only need to gradually make minor adjustments as we deal with more difficult problems. The first setting we consider is the $(1+\epsilon)\Delta$-vertex coloring problem. Our main result in this section is \cref{thm:d1-slack-coloring}:

\begin{theorem}
    \label{thm:d1-slack-coloring}
    Suppose $\Delta \in \Omega(\log^{1+1/\log^*n} n)$.
    There is a {\congest} algorithm that solves the $(1+\epsilon)\Delta$-vertex coloring problem
    w.h.p.\ in $O(\log^* n)$ rounds.
\end{theorem}

Throughout this section, let us assume that all nodes know a common representative family $(S_i)_{i\in[t]}$ with parameters $\alpha=1/2$, $\delta=\frac{\epsilon}{4(1+\epsilon)}$, and $\nu=n^{-3}$ over the color space $U=[(1+\epsilon)\Delta]$. The nodes may, for example, all compute the lexicographically first $(\alpha,\delta,\nu)$-representative family over $U$ guaranteed by~\cref{lem:representative-sets}, with $t \in O(\Delta \cdot n^3)$ and $s \in O(\log n)$, at the very beginning of the algorithm.

We leverage this representative family in a procedure we call \alg{MultiTrials}, where nodes can try up to $\Theta(\log n)$ colors in a round. 
The trade-off 
is that the colors they try are not fully random but picked from a representative set. We show that this does not matter in this application.

\begin{algorithm}
    \caption{Procedure~$\alg{MultiTrials}(x)$ (vertex coloring version)}
    \label[\alg{MultiTrials}]{alg:d1-multitrials}
        
    \begin{enumerate}
        \item $v$ picks $i_v\in [t]$ uniformly at random and chooses a subset $X_v$ of $x$ colors uniformly at random in $S_{i_v}\cap \psi_v$. These are the colors $v$ tries. $v$ describes $X_v$ to its neighbors in $O(1)$ rounds by sending $i_v$ and $(\delta_{\event{c \in X_v}})_{c \in S_{i_v}}$ in $\log(t) +s \in O(\log n)$ bits.
        \item If $v$ tried a color that none of its neighbors tried, $v$ adopts one such color and informs its neighbors of it.
    \end{enumerate}
\end{algorithm}

Using \alg{MultiTrials} with an increasing number of colors, we immediately get an $O(\log^* n)$ algorithm for the $(1+\epsilon)\Delta$-coloring problem (\cref{alg:d1-coloring}).

\begin{algorithm}
    \caption{Algorithm for $(1+\epsilon)\Delta$-vertex coloring (large $\Delta$)}
    \label{alg:d1-coloring}
    \begin{enumerate}
        \item Nodes compute a common $(\alpha,\delta,\nu)$-representative family over $[(1+\epsilon)\Delta]$ guaranteed by \cref{lem:representative-sets}.
        
        \item For $i \in [0..\log^* n]$, for $O(1)$ rounds, each uncolored node runs \alg{MultiTrials}$(2 \knuthupuparrow i)$.\label{step:d1-multitrials}
        \item For $i \in [0..\log^* n]$, each uncolored node runs \alg{MultiTrials}$\parens*{\frac{\epsilon\Delta\cdot \log^{i/\log^*n}n} {2(1+\epsilon)C_c \log n}}$ $O(1)$ times.\label{step:d1-cleanup}
    \end{enumerate}
\end{algorithm}

To show that \cref{alg:d1-coloring} works, we first show that \alg{MultiTrials}, under the right circumstances, is very efficient at coloring nodes (\cref{lem:multitrials-success-prob}). In fact, given the right ratio between slack and uncolored degree, 
as the nodes try multiple colors, they get colored as if each color tried succeeded independently with constant probability.

\begin{lemma}
    \label{lem:multitrials-success-prob}
    Suppose a node $v$ has slack $s(v)\geq \epsilon\Delta$ and $d^*(v)$ uncolored neighbors. Suppose $x \leq \frac{\epsilon}{2(1+\epsilon)}\Delta$. If $x \leq s(v)/2d^*(v)$, then conditioned on an event of high probability $\geq 1-2\nu$, an execution of $\alg{MultiTrials}(x)$ colors $v$ with probability at least $1-2^{-x/4}$, even conditioned on any particular combination of random choices from the other nodes.
\end{lemma}
\begin{proof}
    Consider the representative set $S_{i_v}$ randomly picked by $v$ in the commonly known representative family of parameters $\alpha=1/2$, $\delta=\frac{\epsilon}{4(1+\epsilon)}$, and $\nu=n^{-3}$. We know that $S_{i_v}$ intersects any set of colors $T \subseteq [(1+\epsilon)\Delta]$ of size at least $\delta (1+\epsilon)\Delta$ in $[1/2,3/2] \frac{\card{T}}{(1+\epsilon)\Delta} \card{S_{i_v}} \geq \frac \delta 2 \card{S_{i_v}}$ positions w.h.p.
    
    Let us apply this with $\psi_v$, the set of colors not currently used by neighbors of $v$, and $T_\good$, the set of colors that are neither already used nor tried in this round by nodes adjacent to $v$.
    
    Clearly, $T_\good \subseteq \psi_v$, $\card{\psi_v} = s(v) + d^*(v)$, and $\card{T_\good} \geq s(v) + d^*(v) - x \cdot d^*(v) \geq (s(v)+d^*(v))/2 = \card{\psi_v}/2$. Both sets are of size at least $\delta (1+\epsilon)\Delta$, therefore w.h.p.\ $\card{S_{i_v}\cap T_\good} \geq \frac 1 2 \card{S_{i_v}} \cdot \frac {\card{T_\good}} {(1+\epsilon)\Delta} \geq \frac 1 4 \card{S_{i_v}} \cdot \frac {\card{\psi_v}} {(1+\epsilon)\Delta} \geq \frac 1 6 \card{S_{i_v} \cap \psi_v}$. 
    
    Therefore, assuming that the above holds and that there are at least $x$ colors in $S_{i_v} \cap \psi_v$, when $v$ picks $x$ random colors in $S_{i_v} \cap \psi_v$, the colors picked each have a chance at least $1/6$ of being in $T_\good$. The probability that none of them succeeds is at most $(5/6)^x \leq 2^{-x/4}$. The event that $S_{i_v}$ does not have an intersections of unusual size with either $\psi_v$ or $T_\good$ has probability at least $1-2\nu$.
\end{proof}

The second part of the argument consists of showing that the ratio of slack to uncolored degree increases as \cref{alg:d1-coloring} uses \alg{MultiTrials} with an increasing number of colors. \cref{lem:multitrial-next-phase} helps guarantee that the repeated use of \alg{MultiTrials} leaves all uncolored nodes with an uncolored degree at most $C_c \log n$ for some constant $C_c$.

\begin{lemma}
    \label{lem:multitrial-next-phase}
    Suppose the nodes all satisfy $d^*(v) \leq s(v) / (2 \cdot 2\knuthupuparrow i)$, with $s(v) / (2 \cdot 2\knuthupuparrow i) \geq C_c \log n$. Then after $O(1)$ rounds of $\alg{MultiTrials}(2\knuthupuparrow i)$, w.h.p., they all satisfy $d^*(v) \leq \max(s(v) / (2 \cdot 2\knuthupuparrow (i+1)), C_c\log n)$.
\end{lemma}
\begin{proof}
    Let $v$ be a node of uncolored degree at least $C_c \log n$ (if not, it already satisfies the desired end property). 
    
    By \cref{lem:multitrials-success-prob}, each uncolored neighbor of $v$ stays uncolored with probability at most $2^{-(2\knuthupuparrow i)/4}$. By a Chernoff bound, $C_c$ being large enough, at most $2^{1/4}\cdot2^{-(2\knuthupuparrow i)/4}\cdot d^*(v)$ neighbors of $v$ stay uncolored w.h.p.
    
    Let us repeat this process for $4$ rounds. If at any point the uncolored degree drops below $C_c \log n$, we reached the desired property, and the argument is over. Otherwise, we can apply the Chernoff bound for all $4$ rounds and get that at most $2\cdot2^{-(2\knuthupuparrow i)}\cdot d^*(v) = 2 \cdot\frac 1 {2\knuthupuparrow (i+1)}\cdot d^*(v)$ neighbors of $v$ stay uncolored, so the new uncolored degree of $v$ satisfies:
    \[d^*(v) \leq 2 \cdot\frac 1 {2\knuthupuparrow (i+1)} \cdot \frac {s(v)}{2 \cdot 2\knuthupuparrow i} \leq \frac {s(v)}{2 \cdot 2\knuthupuparrow (i+1)}\ ,\]
    which completes the proof.
\end{proof}

\begin{lemma}
\label{lem:multitrial-clean-up}
    Suppose the nodes all satisfy $d^*(v) \leq C_c \log^{1-i/\log^*n} n$. Then after $O(1)$ rounds of $\alg{MultiTrials}
    \parens*{\frac{\epsilon\Delta\cdot \log^{i/\log^*n} n} {2(1+\epsilon)C_c \log n}}
    $, w.h.p., they all satisfy $d^*(v) \leq C_c \log^{1-(i+1)/\log^*n} n$.
\end{lemma}
\begin{proof}
Let $x=\frac{\epsilon\Delta\cdot \log^{i/\log^*n} n} {2(1+\epsilon)C_c \log n}$ denote the number of colors tried in our application of $\alg{MultiTrials}$. For each uncolored node $v$ we have $x \leq s(v)/2d^*(v)$. By \cref{lem:multitrials-success-prob}, conditioned on a high probability event, each uncolored node stays uncolored with probability at most $2^{-x/4}$, regardless of the random choices of other nodes. 
We set $q = C_c \log^{1-(i+1)/\log^*n}n$. Since $\Delta \geq \log^{1+1/\log^*n} n$ and $x \geq \frac{\epsilon}{2(1+\epsilon)C_c}\log^{(i+1)/\log^*n}n$, we have $q\cdot x \in \Omega(\log n)$. 

Consider $q$ neighbors of a node $v$, $\Theta(1)$ runs of $\alg{MultiTrials}(x)$ leave them all uncolored with probability at most $2^{-\Omega(q\cdot x)}$. The probability that a set of $q$ neighbors stays uncolored is bounded by $d^*(v)^q \cdot 2^{-\Omega(q\cdot x)} = 2^{-\Omega(q\cdot (x-\log\log n))} = 2^{-\Omega(\log n)}$. So, w.h.p., less than $q$ neighbors of $v$ stay uncolored.
\end{proof}

With \cref{lem:multitrials-success-prob,lem:multitrial-next-phase,lem:multitrial-clean-up} proved, we only need a few additional arguments to complete the proof of~\cref{thm:d1-slack-coloring}.

\begin{proof}[Proof of~\cref{thm:d1-slack-coloring}]
\label{proof:d1-slack-coloring}
    Step~\ref{step:d1-multitrials} of \cref{alg:d1-coloring} with $i=0$ creates a situation where the hypotheses of \cref{lem:multitrial-next-phase} hold for $i=1$. The repeated application of \cref{lem:multitrial-next-phase} guarantees that, w.h.p., all nodes are either colored or have uncolored degree $\leq C_c\log n$.
    
In Step~\ref{step:d1-cleanup}, all nodes start with uncolored degree at most $C_c \log n$ and slack at least $\epsilon \Delta$, thus fitting the hypotheses of \cref{lem:multitrial-clean-up}. Its repeated application yields that after the first $\log^*n-1$ first phases of this step, each node is either already colored or tries $\Omega(\log n)$ colors in each run of \alg{MultiTrials}, which colors all remaining nodes w.h.p.
\end{proof}

\paragraph*{Lower $\Delta$ and concluding remarks}

When $\Delta \in O(\log^{1+1/\log^*n} n)$,
 a simple use of the shattering technique~\cite{BEPS16} together with the recent deterministic algorithm of~\cite{GK20} (using $O(\log^2 \calC\log n)$ rounds with $O(\log \calC)$ bits to compute a degree+1 list-coloring of a $n$-vertex graph whose lists are subsets of $\range{\calC}$) is enough to solve the problem in $O(\log^3\log n)$ {\congest} rounds, which combined with our previous $O(\log^*(n))$ algorithm for $\Delta \in O(\log^{1+1/\log^*n} n)$ means there exists an algorithm for all $\Delta$ that solves the $(1+\epsilon)\Delta$ coloring problem in $O(\log^3\log n)$ {\congest} rounds w.h.p.

\begin{theorem}
    \label{thm:d1-slack-coloring-low-delta}
    There is a {\congest} algorithm that solves the $(1+\epsilon)\Delta$-vertex coloring problem in $O(\log^3 \log n)$ rounds w.h.p.
\end{theorem}

\cref{thm:d1-slack-coloring,thm:d1-slack-coloring-low-delta} also hold if instead of having a palette with $\epsilon \Delta$ more colors than vertices have neighbors, thus having slack from the start, we are instead trying to color a $(1-\epsilon)$-locally sparse graph with $(\Delta+1)$ colors. In this case, nodes try a single random color at the very start of the algorithm to generate slack through \cref{prop:slack-lemma}.

\section{Edge coloring}
\label{sec:edge-coloring}

Moving on to the more complicated setting of edge-coloring, we will see that most of what we proved in the previous section is easily adapted to the edge-coloring setting. We first convert the $(1+\epsilon)\Delta$-vertex coloring result to a $(2+\epsilon)\Delta$-edge coloring and then indicate how the number of colors can be reduced to $(2\Delta-1)$. Finally, we show how it can be combined with another edge coloring algorithm \cite{DGP98} to obtain a superfast $(1+\epsilon)\Delta$-edge coloring, for $\Delta \in O(\log^{1+1/\log^*n} n)$.

\subsection{\texorpdfstring{$(2+\epsilon)\Delta$}{(2+ε)Δ}-edge coloring}

\begin{theorem}
    \label{thm:slack-edge-coloring}
    Suppose $\Delta \in O(\log^{1+1/\log^*n} n)$. There is a {\congest} algorithm that solves the $(2+\epsilon)\Delta$-edge coloring problem w.h.p.\ in $O(\log^* n)$ rounds.
\end{theorem}

To prove \cref{thm:slack-edge-coloring}, the most crucial observation is that the elements of the graph trying to color themselves no longer know their palette. In the edge-coloring setting, each of the two endpoints of an edge $e$ only has a partial view of which colors are used by $e$'s neighbors. Communicating the list of colors used at one endpoint of $e$ to the other endpoint is impractical, as it could require up to $\Theta(\Delta\log\Delta)$ bits.
To circumvent this, we introduce a procedure (\alg{PaletteSampling}) for the two endpoints of an edge $e$ to efficiently sample colors in $\psi_e$, the palette of $e$, again using representative sets. The \alg{MultiTrials} procedure is then easily adapted to the edge-setting by making it use \alg{PaletteSampling}, and the same algorithm as the one we had in the node setting works here, simply swapping its basic building block procedure for an edge-adapted variant.

As before (but with a different color space) let us assume throughout this section that all nodes know a common representative family $(S_i)_{i\in[t]}$ with parameters $\alpha=1/2$, $\delta=\frac{\epsilon}{4(1+\epsilon)}$, and $\nu=n^{-3}$ over the color space $U=[(2+\epsilon)\Delta]$.

For each edge $e$, let us denote by $v_e$ and $v'_e$ its two endpoints, with $v_e$ the one with the highest ID of the two. Let us denote by $\psi_e$ the \emph{palette} of $e$, the set of colors unused by $e$'s neighboring edges, and for a node $u$ let $\psi_u$ be the set of colors unused by edges around $u$. For an uncolored edge $e$, $\psi_e = \psi_{v_e} \cap \psi_{v'_e}$.

\begin{algorithm}
    \caption{Procedure~$\alg{PaletteSampling}$ (edge-coloring version)}
    \label{alg:palette-sampling}
    \begin{enumerate}
        \item $v_e$ picks $i_e\in [t]$ uniformly at random and sends $i_e$ to $v'_e$ in $O(\log(t) / \log(n)) = O(1)$ rounds.
        \item $v'_e$ replies with $s$ bits describing $S_{i_e} \cap \psi_{v'_e}$ in $O(1)$ rounds.
        \item $v_e$ sends $s$ bits to $v'_e$ describing $S_{i_e} \cap \psi_{v_e}$ in $O(1)$ rounds.
    \end{enumerate}
\end{algorithm}

\begin{proposition}
\label{prop:palette-sampling}
Suppose $e$'s palette $\psi_e$ satisfies $\card{\psi_e} \geq \delta \cdot (2+\epsilon)\Delta$. Then $v_e$ and $v'_e$ find $[1-\alpha,1+\alpha]\cdot s\cdot\card{\psi_e} / (2+\epsilon)\Delta$ colors in $e$'s palette in an execution of \alg{PaletteSampling} w.h.p.
\end{proposition}
\begin{proof}
    The result follows directly from Equation~\ref{eq:rep-set-large-T} in the definition of representative sets (\cref{def:representative-sets}).
\end{proof}

\alg{PaletteSampling} leverages that while it requires quite a bit of communication for an endpoint of an edge to learn which colors are used at the other endpoint, sending a random color for the other endpoint to reject or approve is quite communication-efficient.
The representative sets and the slack at the edges' disposal further allow us to sample not just $\Theta(\log n/\log \Delta)$ colors (represented in $\log \Delta$ bits each) in $O(1)$ rounds but $\Theta(\log n)$ colors by sampling pseudo-independent colors.

\begin{algorithm}[ht]
    \caption{Procedure~$\alg{MultiTrials}(x)$ (edge-coloring version)}
    \label{alg:multitrials}
        
    \begin{enumerate}
        \item $v_e$ and $v'_e$ execute $\alg{PaletteSampling}$. Let $S_{i_e}$ be the randomly picked representative set.
        \item $v_e$ picks a subset $X_e$ of $x$ colors uniformly at random in $S_{i_e}\cap \psi_e$ and sends $s$ bits to $v'_e$ to describe it. These are the colors $e$ tries.
        
        At this point, each node $u$ knows which colors are tried by all its incident edges.
        \item Each $v_e$ describes to $v'_e$ which of the $x$ colors tried by $e$ were not tried by any other edge adjacent to $v_e$ in $O(1)$ rounds, and reciprocally.
        \item If $e$ tried a color that no edge adjacent to $e$ tried, $v_e$ picks an arbitrary such color, sends it to $v'_e$, and $e$ adopts this color.
    \end{enumerate}
\end{algorithm}

An execution of \alg{MultiTrials} maintains the invariant that each node knows which colors are used by edges incident to it. As before, the representative sets guarantee that for any uncolored edge $e$, whatever colors other edges adjacent to $e$ are trying, the chosen representative set $S_{i_e}$ has a large intersection with the set of unused and untried colors, as long as this set represents a constant fraction of the color space (which slack and a good choice of $x$ guarantee).

\begin{algorithm}[ht]
    \caption{Algorithm for $(2+\epsilon)\Delta$-edge coloring (large $\Delta$)}
    \label{alg:slack-edge-coloring}
    \begin{enumerate}
        \item Nodes send their ID to their neighbors.
        \item Nodes compute a common $(\alpha,\delta,\nu)$-representative family over $[(2+\epsilon)\Delta]$.
        
        \item For $i \in [0..\log^* n]$, for $O(1)$ rounds, each uncolored edge runs \alg{MultiTrials}$(2 \knuthupuparrow i)$,\label{step:multitrials}
        \item For $i \in [0..\log^* n]$, for $O(1)$ rounds, each uncolored edge runs \alg{MultiTrials}$\parens*{\frac{\epsilon\Delta\cdot \log^{i/\log^*n}n} {2(2+\epsilon)C_c \log n}}$.
    \end{enumerate}
\end{algorithm}

\cref{alg:slack-edge-coloring} is exactly the same algorithm as \cref{alg:d1-coloring} in which we have swapped the node version of \alg{MultiTrials} for its edge-variant, which makes for a straightforward proof.

\begin{proof}[Proof of~\cref{thm:slack-edge-coloring}]
\label{proof:edge-coloring}
    The procedure \alg{MultiTrials} adapted to the edge-setting has the same properties as the \alg{MultiTrials} procedure we analyzed in the vertex coloring setting. More precisely, \cref{lem:multitrials-success-prob,lem:multitrial-next-phase} still hold (with the line graph $L(G)$ instead of $G$ and edges instead of nodes), and we can simply refer to the Proof of \cref{thm:d1-slack-coloring} for the details of how all edges get colored w.h.p.\ by \cref{alg:slack-edge-coloring}.
\end{proof}

\subsection{\texorpdfstring{$(2\Delta-1)$}{(2Δ-1)}-edge coloring}

The algorithms we just gave for $(2+\epsilon)\Delta$-edge coloring are easily adapted to the $(2\Delta-1)$-edge coloring setting by creating slack through \cref{prop:line-graph-sparsity,prop:slack-lemma} at the start.

\begin{theorem}
There is a \CONGEST algorithm that solves the $(2\Delta-1)$-edge coloring problem w.h.p.\ in $O(\log^4\log n)$ rounds.
When $\Delta = \Omega(\log^{1+1/\log^*n} n)$, the time complexity is $O(\log^* n)$.
\end{theorem}

What remains is to handle the small-degree case.

\paragraph*{Algorithm for small $\Delta$}

Obtaining an $O(\poly(\log \log n)$-round {\congest} algorithm when 
$\Delta \in O(\log^{1+1/\log^*n} n)$
requires some care in the edge-setting. We sketch how to get an $O(\log^4\log n)$ algorithm here, by first running our $O(\log^*n)$ algorithm to reduce the uncolored degree to $O(\log n)$, and then shattering the graph and simulating the deterministic algorithm of~\cite{GK20} for completing a vertex coloring of the line graph (\cref{lem:simulation-overhead}).

Standard shattering usually assumes that the nodes of the graph (edges, nodes of the line graph in our case) try random colors in their palettes. In the edge-setting, this is clearly problematic as, again, each of the endpoints of an edge only partially know the palette of said edge. Fortunately, shattering is still possible if the nodes can repeatedly use a procedure that colors them with an $\Omega(1)$ probability of success that is independent of the random events that occur at a distance at least $d$ for some constant $d$ (see, e.g., Lemma~3.13 in~\cite{HKMN20_fullversion}, where the shattering technique is similarly adapted to the distance-$2$ setting in which, as in the edge-setting, the nodes do not know their palette).

\begin{lemma}
\label{lem:simulation-overhead}
    The deterministic algorithm of~\cite{GK20} for completing a vertex coloring can be simulated with a $O(\log\log n)$ overhead in the edge-setting on $O(\poly \log n)$-sized components of $O(\poly \log n)$ maximum degree and $O(\log n)$ live degree.
\end{lemma}
\begin{proof}
    Two key properties here:
    \begin{itemize}
        \item since $\Delta \in O(\poly(\log n))$, a color fits in only $O(\log \log n)$ bits.
        \item since the live degree after shattering is at most $O(\log n)$, an edge only needs to receive $O(\log n \log \log n)$ bits to receive one color from each of its neighbors, which only takes $O(\log \log n)$ rounds.
    \end{itemize}
    
    Before actually running the algorithm of~\cite{GK20}, the edges need to learn more colors in their palette than they have neighbors. This is possible in $O(\log \log n)$ rounds. Indeed:
    \begin{itemize}
        \item If the edges have maximum degree $\Delta$ at most $C' \log n$, they may simply learn all the colors used by their neighbors in $O(\log \log n)$ rounds by simple transmission,
        \item If the edges have maximum degree greater than $C' \log n$ with $C'$ a large enough universal constant, they all have sufficient slack to learn more than $C \log n > d^*_L(e)$ colors of their palette in $O(1)$ executions of \alg{PaletteSampling}.
    \end{itemize}
    
    The algorithm of~\cite{GK20} consists of $O(\log N)$ iterations of a $O(\log^2 \calC)$ procedure using messages of size $O(\log \calC)$, where $N$ is the number of nodes in the graph ($O(\poly \log n)$ in the case of our connected components) and $\calC$ is the size of the color space ($(2\Delta-1) \in O(\log^{1+1/\log^*n} n)$ in our case). We simulate this algorithm on the line graph using that $O(\log n)$ messages of size $O(\log \calC)$ can be sent on an edge in $O(\log \calC) \subseteq O(\log \log n)$ {\congest} rounds.
\end{proof}

\subsection{\texorpdfstring{$(1+\epsilon)\Delta$}{(1+ε)Δ}-edge coloring}

Dubhashi, Grable, and Panconesi \cite{DGP98} gave an algorithm for $(1+\epsilon)\Delta$-edge coloring running in $O(\log n)$ rounds of \LOCAL.
Their algorithm has two phases. In the first phase, subsets of edges try random colors from their palette. After the first phase, the maximum (uncolored) degree of a node is at most $\epsilon\Delta/2$. The second phase then applies a $(2\Delta-1)$-edge coloring algorithm with a fresh set of colors. Since the first phase uses $\Delta$ colors, the total number of colors used is at most $(1+\epsilon)\Delta$. 

The first phase runs in $O(1)$ rounds. In \cite{DGP98}, the algorithm used in the second phase runs in $O(\log n)$ rounds, and hence the time bound of their full algorithm. By using the $O(\log^* \Delta)$-round algorithm explained earlier, the total time complexity is reduced to $O(\log^* \Delta)$. What remains is to explain how to implement the first phase in the {\CONGEST} model.

The first phase consists of $t_\epsilon = O(1)$ iterations, where iteration $i$ consists of the following steps: 
Each vertex $u$ randomly selects an $\epsilon/2$-fraction of the edges incident on itself. An edge is considered selected if either of its endpoints select it. Each selected edge $e$ chooses independently at random a tentative color $t(e)$ from its palette (of currently available colors). An edge is assigned its tentative color if no adjacent edge also chose the same tentative color, and the palettes of the edges are updated accordingly.

The only difference in the random selection of the first phase is that only a subset of the edges pick tentative colors. This is easily performed identically in {\CONGEST}. The only issue is then how to pick a random color from within the current palette of the edge. 
We show here how to achieve this approximately using representative sets.

\begin{proposition}
\label{prop:uniform-sampling}
Let an edge $e$ have a palette $\psi_e$ of size $\card{\psi_e}\geq \delta(1+\epsilon)\Delta$. Then in an execution of \alg{PaletteSampling} followed by the edge trying a single color in the sampled palette, conditioned on an event occurring w.h.p., each color of $\psi_e$ gets sampled with a probability in $\range*{\frac{1-\alpha}{1+\alpha},\frac{1+\alpha}{1-\alpha}}\cdot\frac 1 {\card{\psi_e}}$.
\end{proposition}
\begin{proof}
    Let us consider $c \in \psi_e$, a color in $e$'s palette. By Equation \ref{eq:rep-set-single-elem}, the probability that $c$ is in the random representative set $S_{i_e}$ used in \alg{PaletteSampling} is between $\frac{1-\alpha}{(1+\epsilon)\Delta}s$ and $\frac{1+\alpha}{(1+\epsilon)\Delta}s$. Conditioned on $c \in S_{i_e}$, the probability that $c$ is the color that gets picked among the sampled palette colors is $1/\card{S_{i_e} \cap \psi_e}$. When $\card{\psi_e}\geq\delta(1+\epsilon)\Delta$, with probability $\geq 1-\nu$, $\card{S_{i_e} \cap \psi_e} \in [1-\alpha,1+\alpha]\frac{\card{\psi_e}}{(1+\epsilon)\Delta}s$. So conditioned on $\card{S_{i_e} \cap \psi_e}$ being of the  expected order of magnitude, $c$ gets sampled with probability between $\frac{1-\alpha}{1+\alpha}\cdot\frac 1 {\card{\psi_e}}$ and $\frac{1+\alpha}{1-\alpha}\cdot\frac 1 {\card{\psi_e}}$.
\end{proof}

Theorem 5 of \cite{DGP98} shows that the palette sizes and degree of nodes are highly concentrated after each iteration. In particular, each edge has palette of size $\sim (1-p_\epsilon)^{2i}\Delta$, where $p_\epsilon$ is a function of $\epsilon$ alone. 
Thus, in each iteration of phase I, the current palette of each vertex is a constant fraction of $[\Delta]$, and hence \cref{prop:uniform-sampling} applies.

\section{Other applications}
\label{sec:applications} 

\newcommand{\topadA}{\sqrt{\log^{1+1/\log^*n}n}}
\newcommand{\padA}{\vphantom{\topadA^1_\gamma}}
\newcommand{\topadB}{{\log^4\log^3}}
\newcommand{\padB}{\vphantom{\topadB^1_\gamma}}
\newcommand{\topadC}{{\log^{(c)^1}}}
\newcommand{\padC}{\vphantom{\topadC_\gamma}}

\begin{table}[ht]
    \centering
    \begin{tabularx}{\textwidth}{|C{0.95}|C{1.45}|C{0.6}|}
        \hline
        Degree & Tasks & Complexity in {\congest} \\ \hline
        \multirow{2}{*}{$\Delta=\Omega(\log^{1+1/\log^*n}n)$} & $(1+\epsilon)\Delta$-vertex coloring & \multirow{2}{*}{$O(\log^*n)$} \\
        & $(1+\epsilon)\Delta$-edge coloring & \\ \hline
        \multirow{2}{*}{$\Delta=O(\log^{1+1/\log^*n}n)$} & $(1+\epsilon)\Delta$-vertex coloring &  $\padB O(\log^3 \log n)$ \\ \cline{2-3}
        & $(2\Delta-1)$-edge coloring & $\padB O(\log^4 \log n)$ \\ \hline
        $\padA\Delta=\Omega(\sqrt{\log^{1+1/\log^*n}n})$ & \multirow{2}{*}{$(1+\epsilon)\Delta^2$-vertex distance-$2$ coloring} & $O(\log^* n)$\\ \cline{1-1}\cline{3-3}
        $\padA\Delta=O(\sqrt{\log^{1+1/\log^*n}n})$ & & $O(\log^4 \log n)$ \\\hline
        \multirow{2}{*}{$\Delta=\Omega(\log^{1+1/c'}n)$} & $\padC \Delta\log^{(c)}$-vertex coloring & \multirow{3}{*}{$O(1)$} \\
        & $\padC \Delta\log^{(c)}$-edge coloring & \\ \cline{1-2}
        $\padA\Delta=\Omega(\sqrt{\log^{1+1/c'}n})$ & $\Delta^2\log^{(c)}n$-vertex distance-$2$ coloring & \\ \hline
    \end{tabularx}
    \caption{Summary of our results. $\log^{(c)}$ is the $c$-iterated logarithm,  $c$ and $c'$ are constants. Note that an algorithm using $(1+\epsilon)\Delta$ colors implies one using $(2\Delta-1)$, which itself implies one using $(2+\epsilon)\Delta$. The vertex coloring results for large $\Delta$ imply equivalent results with less colors on locally sparse graphs through \cref{prop:slack-lemma}. 
    }
    \label{tab:summary}
\end{table}

Our sampling technique yields a few other interesting results.

The first results are based on the observations that with slack of $\Delta \log^{(c)} n$ (where $\log^{(c)} n$ is the $c$-iterated logarithm), it suffices to run \alg{MultiTrial} for $O(c)$ rounds to reduce the uncolored degree to $O(\log n)$, and that when $\Delta = \Omega(\log^{1+1/c'}n)$, if suffices to run \alg{MultiTrial} for $O(c')$ rounds to color all remaining nodes in the last phase of our algorithms.

\begin{theorem}
\label{thm:delta-log-colors}
$O(\Delta\log n)$-vertex coloring can be done in a single \CONGEST round, w.h.p. 
$(\Delta\log^{(c)} n)$-vertex coloring can be done in $O(1)$ \CONGEST rounds, w.h.p., for any constants $c$ and $c'$, when $\Delta = \Omega(\log^{1+1/c'} n)$. 
\end{theorem}

The \alg{PaletteSampling} and \alg{MultiTrials} procedures are easily adapted to the distance-2 setting, which immediately yields analogue results in this setting. 

\begin{theorem}
\label{thm:delta-log-colors-dist-2}
Distance-2 coloring with $\Delta^2 \log^{(c)} n$ colors can be done in $O(1)$ {\congest} rounds, for any constants $c$ and $c'$, when $\Delta = \Omega(\sqrt{\log^{1+1/c'}n})$, w.h.p.
Distance-2 coloring a graph $G$ with $\Delta^2+1$ colors where $G^2$ is  $(1-\epsilon)$-locally sparse, can be achieved in $O(\log^* n)$ rounds, w.h.p., for $\Delta = \Omega(\sqrt{\log^{1+1/\log^*n}n})$.
\end{theorem}

Indeed, let $\psi^k_v$ denote the set of colors unused in $v$'s distance-$k$ neighborhood. \alg{PaletteSampling} can be done by having each $v$ send $S_{i_v}$ and receive $S_{i_v} \cap \psi^1_u$ to and from each direct neighbor $u$, from which $v$ computes $S_{i_v}\cap \psi^2_v = \bigcap_{u\in N(v)} S_{i_v} \cap \psi^1_u$. \alg{MultiTrials} is similarly easily adapted following the same principle.

\section{Explicit representative sets}

Our proof of \cref{lem:representative-sets} -- the existence of representative sets with appropriate parameters -- is non-constructive, and a natural question is whether we could find an explicit construction with similar parameters. We partially answer this question by remarking that an averaging sampler essentially has all the properties we want, bar one, and known explicit constructions based on expander graphs give the right guarantees (see Theorem 1.3 in~\cite{Healy08}). The output of an averaging sampler is not a set but a multiset (i.e., some elements might appear more than once), but it satisfies properties~\ref{eq:rep-set-large-T} and~\ref{eq:rep-set-small-T} of our definition of representative sets (\cref{def:representative-sets}), which implies that most of results may be obtained with an explicit construction. The notable exception is our result for $(1+\epsilon)\Delta$-edge coloring, which relies on the uniformity of sampling single elements through representative sets (property~\ref{eq:rep-set-single-elem}) which does not seem to immediately hold for the explicit construction mentioned here. A weaker analogue of this property may be proved using the Hitting property of expander walks (Theorem~4.17 in~\cite{Vadhan12}), but how to construct an explicit family of representative sets with the exact properties of \cref{def:representative-sets} and \cref{lem:representative-sets} is an open question.

\section{Conclusions}

We have presented a new technique, inspired by communication complexity, for speeding up \CONGEST algorithms. We have applied it to a range of coloring problems (see \cref{tab:summary} for a summary of our results), but it would be interesting to see it used more widely, possibly with extensions.

We obtained a superfast algorithm in \CONGEST for $(1+\epsilon)\Delta$-edge coloring that holds when $\Delta = \Omega(\log^{1+1/\log^*n} n)$.
It remains to be examined how to deal with smaller values of $\Delta$, which in \LOCAL has been tackled via the  Lov\'asz Local Lemma \cite{EPS15}.

\bibliography{refs}

\end{document}